\documentclass[11pt]{article}

\usepackage[left=1in,right=1in,top=1in,bottom=1.25in]{geometry}

\usepackage{algorithm,algorithmic}
\usepackage{amsmath}
\usepackage{booktabs} 
\usepackage{xspace}
\usepackage{hyperref,url}
\usepackage{amsthm}
\usepackage{amsmath}
\usepackage{graphicx}
\usepackage{amssymb}
\theoremstyle{definition}
\usepackage{microtype,wrapfig} 

 \usepackage[most]{tcolorbox}
\usepackage{empheq}
  
\newtcbox{\mymath}[1][]{%
    nobeforeafter, math upper, tcbox raise base,
    enhanced, colframe=blue!30!black,
    colback=blue!30, boxrule=1pt,
    #1}

\newtheorem{problem}{Problem}

\makeatletter
\newcommand*{\rom}[1]{\expandafter\@slowromancap\romannumeral #1@}
\makeatother

\newcommand{\maj}{\mathrm{maj}}
\newcommand{\hide}[1]{} 

\newtheorem{lemma}{Lemma}
\newtheorem{theorem}{Theorem}

\newcommand{\cc}{\text{Correlation Clustering}\xspace}
\newcommand{\cccc}{\text{2-Correlation-Clustering}\xspace}

\def\e{\epsilon}
\def\hT{\widehat{T}}
 \def\r{\rho}

\def\hD{\widehat{D}}
\def\g{\gamma}
\newcommand{\beql}[1]{\begin{equation}\label{#1}}

\newcommand{\beq}[1]{\begin{equation}\label{#1}}
\newcommand{\eeq}{\end{equation}}

\newcommand{\bfrac}[2]{\left(\frac{#1}{#2}\right)}
\newcommand{\brac}[1]{\left(#1\right)}
\def\a{\alpha}

\newcommand{\field}[1]{\mathbb{#1}} 

\newcommand{\Prob}[1]{\ensuremath{{\bf{Pr}}\left[{#1}\right]}}
\newcommand{\Mean}[1]{\ensuremath{{\mathbb E}\left[{#1}\right]}}

\newcommand{\whp}{\textit{whp}\xspace}

\newcommand{\spara}[1]{\smallskip\noindent{\bf #1}}

\begin{document}
\title{Predicting Positive and Negative Links with Noisy Queries: \\   Theory \& Practice}

\author{
Charalampos E. Tsourakakis\thanks{Boston University, ISI Foundation  \href{mailto:ctsourak@bu.edu}{babis@ctsourak@bu.edu} }
\and
Michael Mitzenmacher\thanks{Harvard University, \href{mailto:michaelm@eecs.harvard.edu}{michaelm@eecs.harvard.edu}}
\and 
Kasper Green Larsen\thanks{Aarhus University, \href{mailto:larsen@cs.au.dk}{larsen@cs.au.dk} }
\and 
Jaros{\l}aw B{\l}asiok\thanks{Harvard University, 
\href{mailto:jblasiok@g.harvard.edu}{jblasiok@g.harvard.edu}}
\and
 Ben Lawson\thanks{ Boston University, 
 \href{mailto:balawson@bu.edu}{balawson@bu.edu}} 
 \and
 Preetum Nakkiran\thanks{Harvard University, 
 \href{mailto:preetum@cs.harvard.edu}{preetum@cs.harvard.edu}}
 \and 
 Vasileios Nakos\thanks{Harvard University, \href{mailto:vasileiosnakos@g.harvard.edu}{vasileiosnakos@g.harvard.edu}}
 }

\hypersetup{
	pdftitle = {Predicting Positive and Negative Links with Noisy Queries: Theory \& Practice},
	pdfauthor = {Tsourakakis et al.}
} 

\maketitle

\begin{abstract}
Social networks involve both positive and negative relationships, which can be captured in signed graphs.  The {\em edge sign prediction problem} aims to predict whether an interaction between a pair of nodes will be positive or negative.  We provide theoretical results for this problem that motivate natural improvements to recent heuristics.

The edge sign prediction problem is related to correlation clustering;  a positive relationship means being in the same cluster.  We consider the following model for two clusters: we are allowed to query any pair of nodes whether they belong to the same cluster or not, but the answer to the query is corrupted with some probability $0<q<\frac{1}{2}$. Let $\delta=1-2q$ be the bias. We provide an algorithm that recovers all signs correctly with high probability in the presence of noise  with $O(\frac{n\log n}{\delta^2}+\frac{\log^2 n}{\delta^6})$ queries. This is the best known result for this problem  for all but tiny $\delta$, improving on the recent work of Mazumdar and Saha \cite{mazumdar2017clustering}.  We also provide an algorithm that performs $O(\frac{n\log n}{\delta^4})$ queries, and uses breadth first search as its main algorithmic primitive. While both the running time and the number of queries for this algorithm are sub-optimal, our result relies on novel theoretical techniques, and naturally suggests the use of edge-disjoint paths as a feature for predicting signs in online social networks. Correspondingly, we experiment with using edge disjoint $s-t$ paths of short length as a feature for predicting the sign of edge $(s,t)$  in real-world signed networks. Empirical findings suggest that the use of such paths improves the classification accuracy, especially for pairs of nodes with no common neighbors.

\end{abstract}

\section{Introduction}
\label{sec:introduction}
With the rise of social media, where both positive and negative interactions take place, signed graphs, whose study was initiated by Heider, Cartwright, and Harary \cite{cartwright1956structural,heider1946attitudes,harary1953notion},
have become prevalent in graph mining.  A key graph mining problem is the {\em edge sign prediction problem}, that aims to predict whether an interaction between a pair of nodes will be positive or negative \cite{leskovec2010predicting,leskovec2010signed}. Recent works have developed numerous heuristics for this task that perform relatively well in practice \cite{leskovec2010predicting,leskovec2010signed}. 

In this work we propose a theoretical model for the edge sign prediction problem that is inspired by active learning \cite{settles2010active}, and the famous {\em balance theory}: ``the friend of my enemy is my enemy'', or ``the enemy of my enemy is my friend'' \cite{cartwright1956structural,easley2010networks,heider1946attitudes,enemy}.
Specifically, we model the edge sign prediction problem as a noisy correlation clustering problem \cite{bonchi2014correlation,mathieu2010correlation,makarychev2015correlation}, where we are able to query a pair of nodes $(u,v)$ to test whether they belong to the same cluster (edge sign $+1$) or not (edge sign $-1$).  The query fails to return the correct answer with some probability $0<q<\frac{1}{2}$.  Correlation clustering is a basic data mining primitive with a large number of applications ranging from social network analysis \cite{harary1953notion,leskovec2010predicting} to computational biology \cite{hou2016new}.  The details of our model follow. 

\spara{Model.} Let $V=[n]$ be the set of $n$ items that belong to two clusters. Set $\sigma:V \rightarrow \{-1,+1\}$, and let  $R = \{v \in V(G): \sigma(v) = -1 \}$ and $B = \{v \in V(G): \sigma(v) = +1 \}$ be the sets/groups of red and blue nodes respectively, where  $0 \leq |R| \leq n$.  For any pair of nodes $\{u,v\}$ define $\tau(u, v) = \sigma(u)\sigma(v) \in \{\pm 1\}$ (i.e., $\tau(u,v) = -1$, if $u$ is reported to be in the different cluster than $v$). The coloring function $\sigma$ is unknown and we wish to recover the two sets $R,B$  by querying pairs of items. (We need not recover the labels, just the clusters.) 
Let $\eta_{u, v} \in \{\pm 1\}$ be iid noise in the edge observations, with $\Mean{\eta_{u, v}}=\delta$ for all pairs $u,v\in V$. The oracle returns 
$$\tilde{\tau}(u,v) = \sigma(u)\sigma(v)\eta_{u, v}.$$
Equivalently, for each query we receive the correct answer with probability $1-q=\frac{1}{2}+\frac{\delta}{2}$, where $q>0$ is the corruption probability.  Our goal is answer the following question.

\begin{tcolorbox}
\begin{problem}
\label{clustering-problem} 
\noindent Can we recover the clusters {\em efficiently} with high probability by performing a  {\em small number of queries}? 
\end{problem}
\end{tcolorbox}

The constraint of querying a pair of nodes {\em only once}  in the presence of noise appears  not only in settings where 
a repeated query is constrained to give the same answer but naturally in more complex settings. For example, in crowd-sourcing applications repeated querying does not help much in reducing errors \cite{mazumdar2016clustering,mazumdar2017clustering,verroios2015entity}, and in biology testing for one out of several millions of potential interactions in the   human protein-protein interaction network involves both experimental noise, and a high cost.

\spara{Main results.} Our two theoretical results show that we can recover the clusters $(R,B)$ with high probability\footnote{An event $A_n$ holds with high probability ({\it whp}) 
if $\lim\limits_{n \rightarrow +\infty} \Prob{A_n}=1$.} in polynomial time. Our first result is stated as the next theorem.

\begin{theorem}
\label{mainthrm1} 
There exists a polynomial algorithm with query complexity $O(\frac{n \log n}{\delta^2}+\frac{\log^2 n}{\delta^6})$ that returns both clusters of $V$ \whp. 
\end{theorem}

Our algorithm improves the current state-of-the-art due to Mazumdar and Saha \cite{mazumdar2017clustering}. Specifically, their information theoretical optimal algorithm that performs $O( \frac{n \log n}{\delta^2})$ queries  requires quasi-polynomial runtime and is unlikely to be improved assuming the planted clique conjecture. On the other hand, their efficient poly-time algorithms require $O(\frac{n \log n}{\delta^4})$ queries.  Our algorithm is optimal for all but tiny $\delta$, i.e., as long as the first term $\frac{n \log n}{\delta^2}$ dominates (asymptotically) the second term $\frac{\log^2 n}{\delta^6}$.

We also provide an additional algorithm that is sub-optimal with respect to both  the number of queries and the  runtime.  Nonetheless, we believe that our algorithm is of independent interest (i) for the novelty of the techniques we develop, and (ii) for the insights that suggest the use of signed edge-disjoint paths as features for predicting whether an interaction between two agents in an online social network will be positive or negative. Our second algorithm is non-adaptive, i.e., it performs all queries upfront, in contrast to our first algorithm. Also, the algorithm itself is simple, using breadth first search as its main algorithmic primitive.   Our result is stated as  Theorem~\ref{thm:thrm1}.  

\begin{theorem}
\label{thm:thrm1}
Let  $\Delta = O(\mathrm{max}\{\frac{1}{\delta^4}\log n,  (\frac{1}{\delta} )^{4 + \frac{2+2\epsilon}{\epsilon} }\})$, and $\epsilon=\frac{1}{\sqrt{ \log \log n}}$. There exists a polynomial time algorithm that performs $\Theta(n \Delta)$ edge queries and recovers the clustering $(R,B)$ \whp for any bias $0< \delta = 1-2q<1$.   
\end{theorem}

Notice that when $\delta$ is constant, asymptotically $O(\frac{n\log n}{\delta^4})$  queries suffice to recover the clustering \whp. Our algorithm is path based, i.e., in order to predict the sign of an edge $(s,t)$, it  carefully creates sufficiently many paths between $s,t$.  While our algorithm (see Section~\ref{sec:proposed} for the details) is intuitive, its analysis involves mathematical arguments that may be of independent interest. Our analysis improves significantly a previous result by the first two authors \cite{mitzenmacher2016predicting}.

Inspired by our path-based algorithm, we use edge-disjoint  $s-t$ paths of short length in a heuristic way to predict the sign of an edge $(s,t)$ in a {\em given} signed network.  Specifically, we perform logistic regression using edge-disjoint  $s-t$ paths of short length as a class of  features in addition to the features introduced in 
\cite{leskovec2010predicting}
to predict positive and negative links in online social networks. Our experimental findings  across a wide variety of real-world signed networks suggest that such paths provide additional useful information to the classifier, with paths of length three being  most informative. The improvement we observe is significantly pronounced for edges with no common neighbors.

\section{Related Work}
\label{sec:related}
\spara{Clustering with Noisy Queries.} Closest to our work lies the recent work of Mazumdar and Saha \cite{mazumdar2017clustering}. Specifically, the authors study Problem~\ref{clustering-problem} in \cite{mazumdar2017clustering} as well, as well as the more general version where the number of clusters is $k\geq 3$. Each oracle query provides a noisy answer on whether two nodes belong to the same cluster or not.  They provide an algorithm that performs $O(\frac{nk\log n}{\delta^2})$ queries, recovers all clusters of size $\Omega( \frac{\log n}{\delta^2})$ where $k$ is the number of clusters, but whose runtime is quasi-polynomial hence impractical, and  unlikely to be improved under the planted clique hardness assumption. They also design a computationally efficient algorithm that runs in $O(n \log n+k^6)$ time and performs $O(\frac{nk^2 \log n}{\delta^4})$ queries. Finally, for  $k=2$ they provide a non-adaptive algorithm that performs $O(\frac{n \log n}{\delta^4})$ and runs in $O(n \log n)$ time.

\spara{Signed graphs.} Fritz Heider introduced the notion of a signed graph 
 in the context of balance theory \cite{heider1946attitudes}.  The key subgraph in balance theory is the {\em triangle}: any set of three fully interconnected nodes whose product of edge signs is negative is not  balanced. The  complete graph is balanced if every one of its triangles
is balanced. Early work on signed graphs focused on graph theoretic properties of balanced graphs  \cite{cartwright1956structural}.  Harary proved the famous balance theorem which characterizes balanced graphs as graphs with two groups of nodes \cite{harary1953notion}. 

\spara{Predicting signed edges.} Since the rise of social media, there has been a surging interest in understanding how users interact among each other.  Leskovec, Huttenlocher, and Kleinberg \cite{leskovec2010predicting}   formulate the edge sign prediction problem as follows:  given a social network $G(V,E)$ with signs on all its edges except for the sign $sgn(x,y)$ on the edge from node $x$ to node $y$, how reliably can we infer $sgn(x,y)$ from the rest of the network? In their original work, Leskovec et al. proposed a machine learning framework to solve the edge sign prediction problem. They trained a logistic regression classifier using 23 features in total. 
Specifically, the first seven features are the following: positive and negative out-degrees $d^+_{out}(x), d^-_{out}(x)$ of node $x$, positive and negative in-degrees   $d^+_{in}(y), d^-_{in}(y)$ of node $y$, the total out- and in-degrees  $d_{out}(x), d_{in}(y)$ of nodes $x,y$ respectively, and the number of common neighbors (forgetting directions of edges) $C(x,y)$ between $x,y$. The quantity $C(x,y)$ was referred to as the {\em embeddedness}  of the edge $x \rightarrow y$ in \cite{leskovec2010predicting}, and we will follow the same terminology.  In addition to these seven features, Leskovec et al. used a 16-dimensional count vector, with one coordinate for each possible triad configuration between $x,y$.   Given a directed edge $(x,y)$ and a third neighbor $v$ connected to both, there are two directions for the edge between $v$ and $x$ and two possible signs for this edge, and similarly for $v$ and $y$, giving 16 possible triads. The 16 possible triads are shown in Table~\ref{tab:triads}.

\small{
\begin{table}[!ht]
\begin{center}
\begin{tabular}{|l|c|c|c|} \hline
Type  & Triad    & Type  & Triad \\ \hline 

1 &$x \xrightarrow[]{+} v, v \xrightarrow[]{+} y$ & 9 &  $x \xrightarrow[]{+} v, v \xleftarrow[]{+} y$ \\ 
2 & $x \xrightarrow[]{+} v, v \xrightarrow[]{-} y$& 10 &  $x \xrightarrow[]{+} v, v \xleftarrow[]{-} y$ \\ 
3 & $x \xrightarrow[]{-} v, v \xrightarrow[]{+} y$& 11 &  $x \xrightarrow[]{-} v, v \xleftarrow[]{+} y$ \\ 
4 &$x \xrightarrow[]{-} v, v \xrightarrow[]{-} y$& 12 &  $x \xrightarrow[]{-} v, v \xleftarrow[]{-} y$ \\ 
5 & $x \xleftarrow[]{+} v, v \xrightarrow[]{+} y$& 13 &  $x \xleftarrow[]{+} v, v \xleftarrow[]{+} y$ \\ 
6 & $x \xleftarrow[]{+} v, v \xrightarrow[]{-} y$& 14 &  $x \xleftarrow[]{+} v, v \xleftarrow[]{-} y$ \\ 
7 & $x \xleftarrow[]{-} v, v \xrightarrow[]{+} y$& 15 &  $x \xleftarrow[]{-} v, v \xleftarrow[]{+} y$ \\ 
8 & $x \xleftarrow[]{-} v, v \xrightarrow[]{-} y$& 16 &  $x \xleftarrow[]{-} v, v \xleftarrow[]{-} y$ \\   \hline
\end{tabular}
\end{center}
\caption{\label{tab:triads} The 16 triads  of     edge $(x \rightarrow y)$. }
\end{table}
}

In the original  work of Leskovec et al. \cite{leskovec2010predicting} the classifier's evaluation is only evaluated on edges whose endpoints have embeddedness at least 25.  However, these kind of thresholds on the embeddedness discard a non-negligible fraction of edges in a graph. For instance, the fraction  of   edges with zero embeddedness is  29.83\%, and 6.23\% in the Slashdot and Wikipedia online social networks (see Table~\ref{tab:datasets}) respectively.  Edges with small embeddedness are ``hard'' to classify, because
triads tend to be a significant feature for sign prediction \cite{leskovec2010predicting}. The lack of common neighbors, and therefore of triads, raises the importance of degree-based features for these edges, and these features are known to introduce some damaging bias, see \cite{chiang2011exploiting} for an explanation.

We will see in Section~\ref{sec:exp}  --perhaps against intuition-- 
that  edge-disjoint paths of length three, may be even more informative than triads. For example, in the Wikipedia social network, if we train a classifier using only triads we obtain 57\% accuracy, and if we train a classifier using only paths of length 3, we obtain 74.06\% accuracy.

\spara{Correlation Clustering.} Bansal et al. \cite{bansal2004correlation} studied Correlation Clustering: given an undirected signed graph  partition the nodes into clusters so that the total number of disagreements is minimized. This problem is NP-hard \cite{bansal2004correlation,shamir2004cluster}. Here, a disagreement can be either a positive edge between vertices in two clusters or a negative edge between two vertices in the same cluster. Note that in \cc the number of clusters is not specified as part of the input. The case when the number of clusters is constrained to be at most two is known as \cccc.

We remark that the notion of {\em imbalance} studied by Harary is  the \cccc cost of the signed graph.  Mathieu and Schudy initiated the study of noisy correlation clustering \cite{mathieu2010correlation}. They develop various algorithms when the graph is complete,  both for the cases of a random and a semi-random model. Later, 
Makarychev,   Makarychev,  and Vijayaraghavan proposed an algorithm for graphs with $O(n\text{poly}\log n)$ edges under a semi-random model \cite{makarychev2015correlation}. For more information on \cc see the recent survey by Bonchi et al. \cite{bonchi2014correlation}.

\spara{Planted bisection model.} The following well-studied bisection model is closely connected to our model. Suppose that there are two groups (clusters) of nodes. A graph is generated as follows: the edge probabilities  are $p$ within each cluster, and $q<p$ across the clusters.  The goal is to recover the two clusters given such a graph. If the two clusters are  balanced, i.e., each cluster has $O(n)$ nodes, then one can recover the clusters \whp, see \cite{mcsherry2001spectral,vu2014simple,abbe2016exact}. 
Hajek, Wu, and Xu proved that when each cluster has $n/2$ nodes (perfect balance), the average degree has to scale as $\frac{\log n}{ (\sqrt{1-q} - \sqrt{q})^2 }$ for exact recovery \cite{hajek2016achieving}. Also, they showed that using semidefinite programming (SDP) exact recovery is achievable at this threshold \cite{hajek2016achieving}.

Notice that if (i) we have two balanced clusters, and (ii) we remove all negative edges from a signed graph generated according to our model, then one can apply such techniques to recover the clusters. We observe that when $\delta \rightarrow 0$ the lower bound of Hajek et al. scales as $O( \frac{\log n}{\delta^2})$.  The techniques we develop in Section~\ref{sec:proposed}  work independently of cluster size constraints.

\spara{Other Techniques.} Chen et al. \cite{chen2014clustering,chen2012clustering}  consider also Model \rom{1} and provide a method that can reconstruct the clustering for random binomial graphs with $O(n \text{poly} \log n)$ edges. Their method exploits low rank properties of the cluster matrix, and requires certain  conditions, including conditions on the imbalance between clusters, see \cite[Theorem 1, Table 1]{chen2012clustering}.  Their method is based on a convex relaxation of a low rank problem. Mazumdar and Saha similarly study clustering with an oracle in the presence of side information, such as a Jaccard similarity matrix
 \cite{mazumdar2016clustering}. 
Cesa-Bianchi et al. \cite{cesa2012correlation} take a learning-theoretic perspective on the problem of predicting signs. They use the correlation clustering objective as their learning bias, and show that the risk of the empirical risk minimizer is controlled by the correlation clustering objective.   Chiang et al. point out that the work of Cand{\`e}s and Tao \cite{candes2006robust} can be used to predict signs of edges, and also provide various other methods, including singular value decomposition based methods, for the sign prediction problem \cite{chiang2014prediction}. The incoherence is the key parameter that determines the number of queries, and is equal to the group imbalance $\tau = \max\limits_{\text{cluster~} C} \frac{n}{|C|}$. The number of queries needed for exact recovery under our Model   is $O(\tau^4 n \log^2{n})$, which is prohibitive when clusters are imbalanced.

\section{Proposed Method}
\label{sec:proposed}
\spara{Pythia2Truth, Theorem~\ref{mainthrm1}.} We describe the algorithm  {\sc Pythia2Truth}  that achieves the guarantees of Theorem~\ref{mainthrm1}. 
The algorithm arbitrarily chooses two node-disjoint sets $A,B \subseteq V$ such that $|A| = \Theta(\frac{\log n}{\delta^2})$ and $|B| = \Theta(\frac{\log n}{\delta^4})$.  Then, it
performs  all possible queries between $A,B$. The total number of queries at this step is $\Theta(\frac{\log^2 n}{\delta^6})$. The algorithm then uses the set of labels $\{ \tilde{\tau}(a,b),\tilde{\tau}(a',b) \}_{b \in B}$ to make a guess $\bar{\tau}(a,a')$ for  $\tau(a,a')$ for each pair $a,a' \in A$. This works as follows: for any given pair $\{a,a'\}$ each $b$ casts a vote $vote(a,a',b)$. Specifically,  $vote(a,a',b)=+1$ if $\tilde{\tau}(a,b)=\tilde{\tau}(a',b)$, and $vote(a,a',b)=-1$ if $\tilde{\tau}(a,b) \neq \tilde{\tau}(a',b)$. The prediction $\bar{\tau}(a,a')$ is $+1$ if the majority of votes $\{ vote(a,a',b) \}_{b \in B}$ is $+1$, and $-1$ otherwise.  

The aforementioned steps ensure that $\bar{\tau}(a,a')  = \tau(a,a')$ for all pairs $a,a' \in A$ \whp. Clearly, there exist at least $\Theta(\frac{\log n}{\delta^2})$ nodes from at least one of the two clusters. This set of nodes is found by finding the largest connected component (that is actually a clique) of the graph induced by the positive edges in $A$. This set $C$ serves as a {\em seed set}. For each node $u \notin C$ we perform all queries $(u,c)$ for each $c \in C$. If the majority of the oracle answers is $+1$ then we add $u$ in $C$. The procedure outputs $C$ and its complement as the true clusters. Now we prove the correctness of our proposed algorithm. First, we prove the following lemma. 

\begin{lemma}
\label{lem1main}  
Let $S\subseteq V$ such that $|S| =  \frac{24 \log n}{\delta^4}$. Consider any pair of nodes $u,v \in V\backslash S$, and let $\bar{\tau}(u,v) =$ majority$( \{ \tilde{\tau}(u,s) \cdot \tilde{\tau}(v,s) \}_{s \in S} )$. Then, $\bar{\tau}(u,v) = \tau(u,v)$ with probability at least $1-\frac{1}{n^3}$. 
\end{lemma} 

\begin{proof} 
Consider any pair of nodes $u,v \in  V\backslash S$, and let  $X_s(u,v)$ be an indicator random variable for $s \in S$ that is equal to 1 if  the product $\tilde{\tau}(u,s) \cdot \tilde{\tau}(v,s)$ of the two noisy labels $ \tilde{\tau}(u,s), \tilde{\tau}(v,s)$ is the true label  $\tau(u,v)$. Then, 
$ \Prob{ X_s=1 } = (1-q)^2+q^2 = \frac{1+\delta^2}{2}.$ 
 For notation simplicity let $p=\Prob{ X_s=1 }$. Also, we define $X(u,v) = \sum_{s \in S} X_s(u,v)$. Notice that $\bar{\tau}(u,v) = \tau(u,v)$ iff $X(u,v) \geq \frac{|S|}{2}$. Using Chernoff bounds \cite{mitzenmacher2005probability}, we obtain that the probability of misclassification is bounded by

\begin{align*}
\Prob{ X(u,v) < \frac{|S|}{2} } &= \Prob{ X(u,v) < \frac{p|S|}{2p} } = \Prob{ X(u,v) < \bigg( 1- (1-\frac{1}{2p}) \bigg) p|S| } \\    
&\leq \exp\Big( - \frac{(2p-1)^2}{8p^2} \frac{24 \log n}{\delta^4} p\Big) = \exp\Big( - \frac{\delta^4}{4(1+\delta^2)} \frac{24 \log n}{\delta^4} \Big) <\frac{1}{n^3}. 
\end{align*}
\end{proof}

\begin{algorithm*}[t]
\caption{\label{alg1b}  {\sc Pythia2Truth}($V$) } 
 \begin{algorithmic}  
\STATE Choose arbitrarily $A, B \subseteq V$   two disjoint sets of nodes, such that $|A| = \frac{48\log n}{\delta^2}$, and $|B|= \frac{24 \log n }{\delta^4}$. 
\STATE  Perform all $\Theta(\frac{\log^2 n}{\delta^6})$ queries among $A,B$. 
\FOR{each pair  $a,a' \in A$}
\STATE $\text{counter}_{a,a'} \leftarrow 0$ 
\FOR{each $b \in B$}  
\IF{ $\tilde{\tau}(a,b) = \tilde{\tau}(a',b)$ }
\STATE $\text{counter}_{a,a'} \leftarrow \text{counter}_{a,a'}+1$ 
\ENDIF
\ENDFOR
\IF{$\text{counter}_{a,a'} \geq \frac{|B|}{2} $}
\STATE  $\bar{\tau}(a,a') = +1$  
\ELSE 
\STATE  $\bar{\tau}(a,a') = -1$  
\ENDIF  
\ENDFOR  
\STATE Remove the negative edges from $A$, and let $C$ be the largest clique
\FOR{each $u \in V \backslash C$} 
\STATE Perform all queries $(u,c)$ for $c \in C$ 
\IF{the majority of answers is $+$} 
\STATE $C \leftarrow C \cup \{u\}$
\ENDIF
\ENDFOR
\RETURN $(C,V \backslash C)$
\end{algorithmic}
\end{algorithm*}

A straight-forward corollary of Lemma~\ref{lem1main} derived by taking a union bound over all pairs of nodes in $V \backslash S$ is that our algorithm predicts the labels of all such interactions correctly \whp. Using Lemma~\ref{lem1main} we are also able to prove the correctness of our Algorithm. 

\begin{proof}[Proof of Theorem~\ref{mainthrm1}]
Using lemma~\ref{lem1main} by setting $S = B$ we obtain that all pairwise interactions within the set $A$ are correctly labeled with high probability. By the pigeonhole principle, since $|A| = \frac{48 \log n}{\delta^2}$, one of the two clusters has at least  $\frac{24 \log n}{\delta^2}$ nodes in $A$. This set can easily be found: since within $A$  all labels $\bar{\tau}(a,a')$  are equal to $\tau(a,a')$, for $a,a' \in A$,   disregarding the negative labels $\bar{\tau}(a,a')$ will result in at most two connected cliques. We can find the largest such clique in $O(|A|)$ time (since one step of BFS finds all other nodes). Let $C$ be the corresponding set of nodes. 

Let $u \in V \backslash C$. We perform all possible $|C|$  queries  between $u$ and $C$, and we decide that $u$ belongs to $C$ if the majority of the oracle answers is $+1$. Define $X_c(u)$ to be an indicator random variable that is equal to 1 if the oracle answer for the pair $\{u,c\}$ is correct, and 0 otherwise. Let $X(u) = \sum_{c \in C} X_c(u)$ be the random variable distributed according to $Bin(|C|,1-q)$. The probability of failure is bounded by 
 
\begin{align*}
\Prob{ X(u) < \frac{|C|}{2} } & = \Prob{ X(u) < \bigg( 1- (1-\frac{1}{2(1-q)}) \bigg) (1-q)|C| } \\
&\leq \exp\Big( -\frac{\delta^2}{2(1+\delta)^2} \frac{24 \log n}{\delta^2} \frac{1+\delta}{2} \Big) < \frac{1}{n^3}.\\ 
\end{align*}

\noindent By combining the above results, and a union bound our proposed algorithm succeeds \whp to recover both clusters.
\end{proof}
 
The total runtime of our method is $O\bigg( 
\underbrace{ {\frac{48 \log n}{\delta^2} \choose 2} \frac{24 \log n}{\delta^4}}_{\text{classify all pairs in~}A} + \underbrace{\frac{48 \log n}{\delta^2}}_{\text{find largest clique}}+ \underbrace{\frac{ n\log n}{\delta^2}}_{\text{decide the rest of cluster nodes}}  \bigg)$ that simplifies to $O(\frac{ n \log n}{\delta^2}+\frac{\log^3 n}{\delta^8})$.

\begin{wrapfigure}{r}{0.45\textwidth}
\centering
\includegraphics[width=0.4\textwidth]{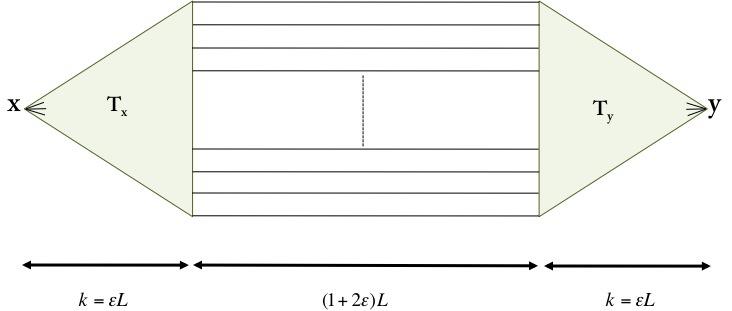} \\
\caption{\label{fig:fig1}  We create for each pair of nodes $x,y$ two node disjoint trees $T_x,T_y$ of depth $k= \epsilon L$  whose leaves can be matched via a natural isomorphism and linked with edge disjoint paths of length $(1+o(1))L$ (see Section~\ref{sec:proofs} for the details).}
\end{wrapfigure}

\spara{A path-based approach, Theorem~\ref{thm:thrm1}.}  Before we go into mathematical details (cf. Section~\ref{sec:proofs}), we describe how our algorithm behind Theorem~\ref{thm:thrm1} works. We perform  $O(n\Delta)$ queries uniformly at random
to predict all possible  ${n \choose 2}$ edge signs under our model,
as called in Theorem~\ref{thm:thrm1}.  Let $G$ be the resulting graph. 
To predict the sign of the node pair $\{x,y\}$, our algorithm   performs --at high level--  two steps. First, we construct a subgraph $G_{x,y}(V_{x,y},E_{x,y}) \subseteq G$. This subgraph is constructed using breadth first search (BFS), and consists of  two isomorphic trees $T_x,T_y$, each one rooted at $x,y$ respectively.  The leaves of these trees can be matched and linked with edge disjoint paths;  more details are given in Section~\ref{sec:proofs}.  Pairs of nodes that map to each other under the isomorphism are written as $v,\bar{v}$, so $y$ is also $\bar{x}$.  The isomorphic copies of the leaves $u \in T_x, \bar{u} \in T_y$ of the two trees are connected by edge disjoint paths. This subgraph is shown in Figure~\ref{fig:fig1}.

Given the subgraph $G_{x,y}$, our algorithm estimates the relative coloring of pairs of nodes recursively, working from the leafs of the trees $T_x, T_y$ up to the roots. That is, we first estimate $\sigma(u)\sigma(\bar u)$ for the leaves $u, \bar{u}$ based on the path between them, and then, moving toward the roots $x$ and $y$, we estimate $\sigma(v)\sigma(\bar v)$ based on a majority vote derived by the children.  More formally, let $Z_{v, \bar v}$ be the estimate of $\sigma(v)\sigma(\bar v)$ for any vertex $v$ given by the algorithm below. (Formally, this algorithm defines the random variables $Z_{u, \bar u}$).

\begin{itemize}
\item \underline{Base case}: For leaf nodes $u \in T_x$, we define
$$Z_{u, \bar u} := \tau(P_i)$$
where $\tau(P_i)$ is our estimate of $\sigma(u)\sigma(\bar u)$ based just on observations from the path $P_i$ from $u \to \bar u$ (that is,
$\tau(P_i) := \prod_{e \in P_i}\tau(e)$).
\item  \underline{Induction on depth}: For nodes $u$ at depth $\ell$ in $T_x$,
let $N(u)$ be children of $u$ (at depth $\ell+1$).
Then, define
$$Z_{u, \bar u} :=
\mathrm{majority}( \{ \tau(u, v) Z_{v, \bar v} \tau(\bar v, \bar u) \}_{v \in N(u)} ).$$
\end{itemize}

Our induction approach collapses each path between each pair of nodes $v,\bar{v}$ (that are children of $u, \bar{u}$ respectively) at depth $\ell+1$ into a single edge, which we estimate based on our previous estimates $Z_{v, \bar v}$. Then, in this ``collapsed'' graph, we take the $\mathrm{majority}$ vote over all (disjoint) paths $u \to \bar u$.
At the end, we output $Z_{x, y} :=Z_{x, \bar x} $.
Using Fourier analytic techniques \cite{o2014analysis} we  prove in Section~\ref{sec:proofs} that $$\Pr[Z_{x, y} = \sigma(x)\sigma(y)] \geq 1-1/n^3.$$  A union bound over all ${n \choose 2}$ pairs yields Theorem~\ref{thm:thrm1}.  Observe that algorithmically we do not need to perform all ${n \choose 2}$ queries to recover the two clusters, but any set of $n-1$ queries that form a spanning tree between the $n$ nodes.

\spara{A machine learning formulation.}  Our algorithm    is heavily based on paths to predict the sign of $\{x,y\}$.  Inspired by this result, we use paths  as an informative feature in the context of predicting positive and negative links in online social networks. Specifically, we enrich the machine learning formulation proposed by Leskovec et al. \cite{leskovec2010predicting} by adding 
four new global features as follows: for each edge $(u,v)$, we find a number of
edge-disjoint paths of length three that connect $u,v$, and similarly we find edge-disjoint paths of length four. 
 We calculate the product of the weights of each path and tally the number of positive and negative products for each path length. 
We add these four counts as four new dimensions. (We also tried paths of length five, but they are not as informative and are also more computationally expensive, so we do not study such paths henceforth.)  We ignore directions of edges both for computational efficiency, and in order to avoid introducing too many features, as for a path of length $\ell$ there are $2^{\ell}$ possible directed versions of the path. We describe some key elements of the   framework in \cite{leskovec2010predicting} for completeness.  Notice that the feature engineering is performed on the whole graph.

\begin{itemize}
\item[-] {\it Features:} In addition to our four new global features, we use  23 local features  to predict the  sign of the edge $u \rightarrow v$: $d^+_{out}(u),  d^-_{out}(u)$, $d^+_{in}(v), d^-_{in}(v), d_{out}(u), d_{in}(v)$, $C(u,v)$ where $C(u,v)$ is the embeddedness, i.e., the number common neighbors of $u,v$ (in an undirected sense), and a 16-dimensional count vector, with one coordinate for each possible configuration of a triad. 
\item[-] We train a logistic regression classifier that learns a model of the form  $ \Prob{ + | x } = \frac{1}{1+e^{-b_0+\sum b_i x_i}}$.    Here $x=(x_1,\ldots,x_{27})$ is our 27-dimensional feature vector. 
\item[-]  We create balanced datasets so that random guessing results in 50\% accuracy. We perform 10-fold cross validation, i.e., we
create 10 disjoint folds, each consisting of 10\% of the total number of edges. For each fold, we use the remaining 90\% of the edges as the training dataset for the logistic regression. We report average accuracies over these 10 folds. 
\end{itemize}

\section{Experimental Results}
\label{sec:exp}
\subsection{Experimental Setup}

\spara{Experimental setting.}  Since finding the maximum number of edge-disjoint paths of short length is NP-hard \cite{itai1982complexity}, we implement a fast greedy  heuristic: to find edge-disjoint paths of length $k$ ($k=3,4$ in our experiments)  between $s,t$, we discard edge directionality, and we start BFS from $s$. As soon as we find a path of length $k$ to $t$, we check if its edges have been removed from the graph using a hash table; if not, we add the path to our collection, we remove its edges from the graph, we add them to the hash table, and we continue. At termination, we count how many positive and negative paths exist in our collection. To train a classifier, we use logistic regression. For this purpose we use  {\em Scikit-learn} \cite{pedregosa2011scikit}.

 \spara{Datasets.} Table~\ref{tab:datasets} shows various publicly available online social networks (OSN)  we use in our experiments together with  the number of nodes $n$ and the number of edges $m$.  We present in detail our findings for the first two datasets described in the following. The results for the other graphs are very similar.  
  
{\em Slashdot} is a news website. Nodes correspond to users, and edges to their interactions. A positive sign means that a user likes another user's comments.  
 
{\em Wikipedia} is  a free online encyclopedia, created and edited by volunteers around the world. 
Nodes correspond to editors, and  a signed link indicates a positive or negative vote by one user on the promotion of another. 

\spara{Machine specs.} All experiments run on a laptop with 1.7 GHz Intel Core i7 processor and 8GB of main memory.

\spara{Code.} Our code was written in Python. A demo of our code is available as a Python notebook online  at \href{https://github.com/tsourolampis/Predicting-Signed-Edges/blob/master/Prediction.ipynb}{github/Prediction.ipynb}. 

\begin{table}[!ht]
\begin{center}
\begin{tabular}{|l|c|c|c|} \hline
Name  & $n$  & $m$  & Description \\ \hline 
{\sc  Slashdot (Feb. 21) }   & 82\,144  & 549\,202 & OSN \cite{snap}\\
{\sc Wikipedia}   &  7\,118 	& 103\,747 & OSN \cite{snap}	 \\ 
{\sc Epinions} & 119\,217 & 841\,200 & OSN \cite{snap} \\ 
{\sc  Slashdot (Nov. 6) }   & 77\,350  & 516\,575 & OSN \cite{snap}\\
{\sc  Slashdot (Feb. 16) }   & 81\,867  & 545\,671 & OSN \cite{snap}\\ 
{\sc Highlands tribes} & 16 & 58 & SN \cite{konect:kenneth54} \\ 
\hline
\end{tabular}
\end{center}
\caption{\label{tab:datasets} Datasets used in our experiments. }
\end{table}

\begin{figure}[t]
\centering
\begin{tabular}{cc} 
\includegraphics[width=0.60\textwidth]{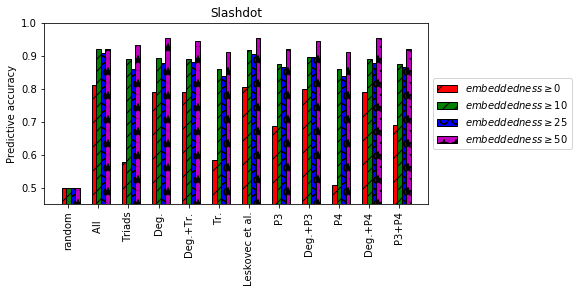} \\
(a) \\  
\includegraphics[width=0.60\textwidth]{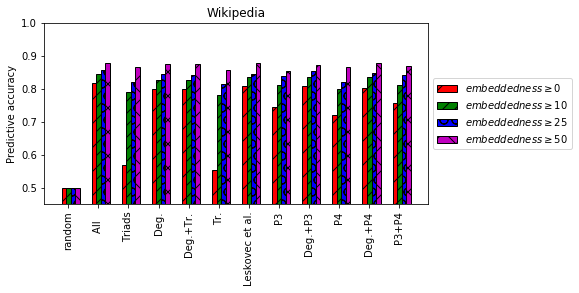}   \\ 
(b)  
\end{tabular}
\caption{\label{fig:overall_performance} Average accuracy of predicting edge signs using 10-fold cross validation. (a) Slashdot, (b) Wikipedia.  }
\end{figure}

\begin{figure*}[!ht]
\centering
\begin{tabular}{@{}c@{}@{\ }c@{}@{}c@{}} 
\includegraphics[width=0.32\textwidth]{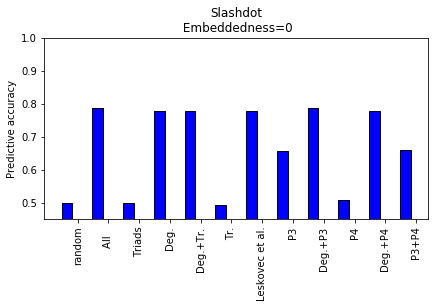}  &
\includegraphics[width=0.32\textwidth]{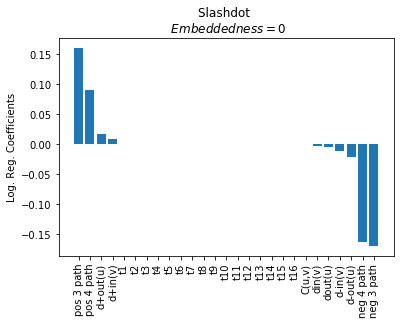} & 
\includegraphics[width=0.32\textwidth]{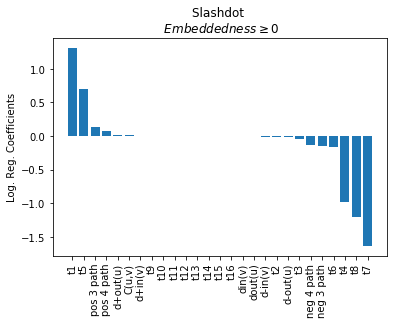} 
\\
(a) & (b)  & (c) \\
\includegraphics[width=0.32\textwidth]{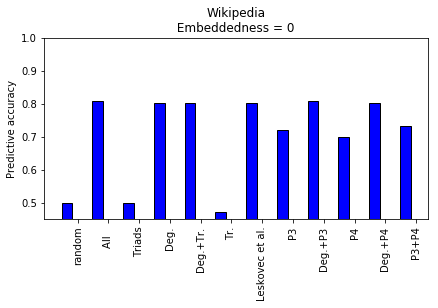} 
 &
\includegraphics[width=0.32\textwidth]{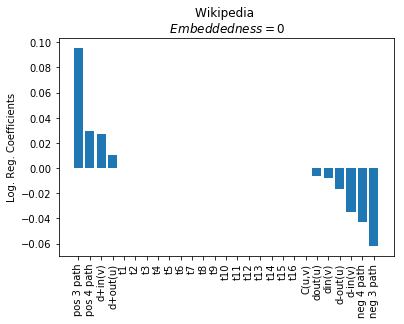} 
& 
\includegraphics[width=0.32\textwidth]{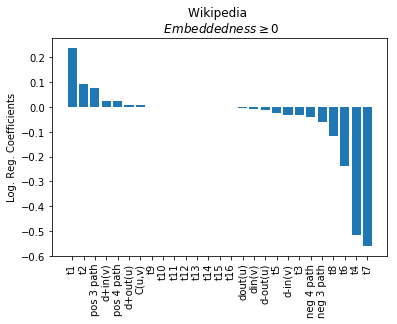} 
\\
(d) & (e) & (f)  \\
\end{tabular}
\caption{\label{fig:coefficients} (a),(d) Average accuracy of predicting  signs of edges with zero embeddedness using 10-fold cross validation, (b),(e) and the resulting logistic regression coefficients, for Slashdot and Wikipedia respectively.  (c), (f) Learned logistic regression coefficients for the whole Slashdot and Wikipedia  datasets respectively.}
\end{figure*}

\subsection{Empirical findings} 
\label{sec:osn}
 
We experiment with various combinations of the 27 features that we described in Section ~\ref{sec:proposed}. {\em All} refers to using all 27 features, {\em Triads} to the 16-dimensional vector of triad counts, {\em Deg} to degree features, {\em Tr.} (short for triangles) to the number of common neighbors, {\em Leskovec et al.} to the  23 features  used in \cite{leskovec2010predicting}, and {\em P3}, {\em P4} to the number of negative and positive edge-disjoint paths of length 3, 4 respectively. A combination of the form  {\em P3+P4} means using the union of these features, for example counts of positive and negative edge disjoint paths of length 3 and 4 respectively. 

Figures~\ref{fig:overall_performance}(a), (b) shows the  performance of our classifier using different combinations of features, broken down by a lower bound on the embeddedness. For the Slashdot dataset,  we observe that when we classify all edges (embeddedness $\geq 0$)  {\em P3} performs better than {\em Triads}, i.e.,  68.8\% vs 57.8\%. Also, the performance of  a {\em Triads}-based classifier is not monotonic as a function of the embeddedness lower bound. For example, when embeddedness is at least 10 the accuracy is 88.9\%, whereas when it is at least 25 it becomes 86.1\%. However, in general the prediction problem 
becomes easier as the embeddedness increases. Also, using all features, i.e., the addition of the four new features {\em P3}, {\em P4} to the existing {\em Leskovec et al.} results in the best possible performance. 
Finally, paths of length 3 are more informative than paths of length 4. This is clearly seen by the logistic regression coefficients shown in  Figure~\ref{fig:coefficients}(c). We also observe that different types of triads can have significantly different regression coefficients, and that the coefficients depend significantly on the graph, as seen in Figures~\ref{fig:coefficients}(c) and~\ref{fig:coefficients}(f).

Figures~\ref{fig:coefficients}(a), (d) show  the average accuracy of predicting edge signs for edges with embeddedness equal to zero for the Slashdot and Wikipedia datasets respectively. When we use {\em Triads} the predictive accuracy is as only about as good as random guessing, i.e., 50\%.   {\em P3} results in 65.74\%, and 71.96\% accuracy, {\em P4} in 50.78\%, and 69.90\% accuracy for Slashdot and Wikipedia respectively.  We observe that using all features leads to the best possible performances of 78.63\%, and 80.92\% accuracy respectively for the two datasets.  The importance of paths of length 3, and 4 for  edges with zero embeddedness is seen by the logistic regression coefficients in  Figures~\ref{fig:coefficients}(b), (e).

\section{Algorithmic Analysis}
\label{sec:proofs}
We use the following notation. Let $\epsilon := \frac{1}{\sqrt{\lg\lg n}}$, 
and 
$$\Delta = O\big(\mathrm{max}\{\frac{1}{\delta^4}\log n,  (\frac{1}{\delta} )^{4 + \frac{2+2\epsilon}{\epsilon} }\}\big)$$ 

\noindent be the average degree. We perform in total 
$\frac{12n\log{n}}{\delta^4}$ queries, and for simplicity,  let the bias $0 < \delta < \frac{1}{2} $ be a constant, independent of $n$. Hence, asymptotically $\Delta =   \frac{12}{\delta^4}\log n $.  Finally,  let $L = \frac{\log n}{\log \Delta}$ be the diameter of the resulting random graph we obtain \whp  \cite{bollobas1998random}.

\begin{algorithm}[t]
\caption{\label{alg:edgeDisjointPaths}  Almost-Edge-Disjoint-Paths($u,v$)} 
 \begin{algorithmic}  
\REQUIRE $G(V,E)$, $u,v \in V(G)$ 
\STATE $\epsilon \leftarrow \frac{1}{\sqrt{\log\log{n}}}$
\STATE Using Breadth First Search (BFS) grow a tree $T_u$ starting from $u$ as follows.  
\STATE We use a branching factor equal to $\frac{4\log{n}}{\delta^4}$ until it reaches depth equal to $\epsilon L$. Similarly, grow a tree $T_v$ rooted at $v$, node disjoint from $T_u$ of equal depth.
\STATE From each leaf $u_i$ ($v_i$) of  $T_u$ ($T_v$)  for $i=1,\ldots, N$ 	grow  node disjoint trees until they reach depth $(\frac{1}{2}+\epsilon)L$ with branching factor $\frac{4\log{n}}{\delta^4}$.  Finally, find an edge between $T_{u_i},T_{v_i}$ 
\end{algorithmic}
\end{algorithm}

\subsection{Subgraph construction} 
\label{sec:subgraph}  

The next lemma follows from standard Chernoff bounds (and a union bound over vertices).
\begin{lemma} 
\label{lem1} 
Let $G \sim G(n,\frac{12\log{n}}{\delta^4 n})$ be a random binomial graph. Then \whp all vertices have degree greater than $\frac{5\log{n}}{\delta^{4}}$. 
\end{lemma} 


\noindent Now we proceed to our construction of sufficiently enough almost edge-disjoint paths. Our construction is based on  standard techniques in random graph theory \cite{broder1998optimal,dudek2015rainbow,frieze2012rainbow,tsourakakis2013mathematical}, we include the full proofs for completeness. 

\begin{lemma}
\label{lem2} 
Let $G\sim G(n,p)$ where $p=\frac{12\log n}{\delta^4n}:= \frac{c_{\delta} \log n}{n}$.  Fix $t\in \field{Z}^+$ and $0<\alpha<1$. Then, \whp there does not exist a subset $S \subseteq [n]$, such that $|S| \leq \alpha t L$ and $e[S] \geq |S|+t$. 
\end{lemma}

\begin{proof}
Set $s=|S|$.Then,
\begin{align*}
\Prob{ \exists S: s \leq \alpha t L \text{~~and~~} e[S] \geq s+t } 
&  
\leq \sum_{s \leq \alpha t L} \binom{n}{s} \binom{\binom{s}{2}}{s+t} p^{s+t} \leq \\ 
  \sum_{s \leq \alpha t L } \bfrac{ne}{s}^s \bfrac{es^2p}{2(s+t)}^{s+t}   
&\leq \sum_{s \leq \alpha t L } (e^{2+o(1)}\log{n})^s \bfrac{\tfrac{c_{\delta}}{2}es\log{n}}{n}^t \leq \\
  \alpha tL \brac{ (e^{2+o(1)}\log{n})^{ \alpha L} \bfrac{\tfrac{c_{\delta}}{2} e \alpha t\log^2{n} }{n\log{\log{n}}}}^t &<   \frac{1}{n^{(1-\a-o(1))t}}. 
\end{align*}
\end{proof}

\begin{lemma}
\label{lem3}
Let $T$ be a rooted (subgraph) tree of depth at most $\frac{4L}{7}$ and let $v$ be a vertex not in $T$.  Then with probability $1-o(n^{-3})$, $v$ has at most $10$ neighbors in $T$, i.e., $|N(v) \cap T| \leq 10$.
\end{lemma} 

\begin{proof}
Let $T$ be a rooted tree of depth at most $\frac{4L}{7}$ and let $S$ consist of $v$, the neighbors of $v$ in $T$ plus the ancestors of these neighbors. Set $b =|N(v) \cap T|$. Then 
$|S|\leq 4bL/7+1\le 3bL/5$ and $e[S]=|S|+b-2$. It follows from Lemma \ref{lem2} with $\a=3/5$ and $t=8$, that 
we must have $b\leq 10$ with probability $1-o(n^{-3})$.
\end{proof}

\noindent We show that by growing trees iteratively we can construct 
sufficiently many edge-disjoint paths for $n$ sufficiently large.

\begin{lemma}
\label{lem4} 
Let $k=\epsilon L$. For all pairs of  vertices $x,y \in [n]$ there exists a subgraph $G_{x,y}(V_{x,y},E_{x,y})$ of $G$ as shown in Figure~\ref{fig:fig1}, {\it whp}.
The subgraph consists of two isomorphic vertex disjoint trees $T_x,T_y$ rooted at $x,y$ each of depth $k$.
$T_x$ and $T_y$ both have a branching factor of  $\frac{4\log n}{\delta^4} $. 
If the leaves of $T_x$ are $x_1,x_2,\ldots,x_\tau,\tau\geq n^{4\epsilon/5}$ then $y_i=f(x_i)$ where $f$ is a natural 
isomorphism.
Between each pair of leaves $(x_i,y_i),i=1,2,\ldots,m$ 
there is a path $P_i$ of length  $(1+2\epsilon) L$. The paths $P_i,i=1,2,\ldots,\tau,\ldots$ are edge disjoint.
\end{lemma}

\begin{proof}
Since we have to do this for all pairs $x,y$, we note without further comment that likely (resp. unlikely) 
events will
be shown to occur with probability $1-o(n^{-2})$ (resp. $o(n^{-2}$)).

To find the subgraph shown in Figure~\ref{fig:fig1} we grow tree structures as shown 
in Figure~\ref{fig:fig2}.
Specifically, we first grow a tree from $x$ using BFS until it reaches  depth $k$. 
Then, we grow a tree starting from $y$ again using BFS until it reaches depth $k$.  Finally, once trees $T_x,T_y$ have been constructed, we grow trees from the leaves of $T_x$ and $T_y$ using BFS for depth $\gamma=(\frac{1}{2}+\epsilon)L$. 
We analyze these processes, explaining
in detail for $T_x$ and outlining the differences for the other trees. 
We use the notation $D_i^{(\r)}$ for the number of vertices at depth $i$
of the BFS tree rooted at $\r$.

First we grow $T_x$. As we grow the tree via BFS from a vertex $v$ at depth $i$ to vertices 
at depth $i+1$ certain {\em bad} edges from $v$ may point 
to vertices already in $T_x$. Lemma \ref{lem3} shows with probability $1-o(n^{-3})$ there can be at most 10 bad edges emanating from $v$.

Hence, we obtain the recursion  
\beq{rec1}
D_{i+1}^{(x)} \geq \brac{ \frac{5\log{n}}{\delta^4}-10} (D_i^{(x)}-1) \geq \frac{4 \log{n} }{\delta^4}D_i^{(x)}.
\eeq

\noindent Therefore the number of leaves satisfies 

\beq{rec2}
D_{k}^{(x)} \geq   \Big( \frac{4\log n}{\delta^4} \Big)^{\e L}\geq n^{4\e/5}.
\eeq

\noindent We can make the branching factors exactly $\frac{4\log n}{\delta^4}$ by pruning. We do this so that the trees $T_x, T_y$ are isomorphic to each other. With a similar argument $D_{k}^{(y)} \geq n^{\frac{4}{5}\epsilon}$. Specifically, the only difference is that now we also say an edge is bad if the other endpoint is in $T_x$.
This immediately gives
$$D_{i+1}^{(y)} \geq \brac{ \frac{5\log{n}}{\delta^4}-20} (D_i^{(y)}-1) \geq \frac{4\log{n}}{\delta^4} D_i^{(y)}$$
and the required conclusion.  

Similarly, from each leaf $x_i \in T_x$ and $y_i \in T_y$ we grow trees $\hT_{x_i},\hT_{y_i}$ of depth 
$\gamma = \big(\frac{1}{2}+\epsilon\big) L$ using the same procedure and arguments 
as above. Lemma~\ref{lem3} implies that there are at most 20 edges from the vertex $v$ being explored to 
vertices in any of the trees already constructed (at most 10 to $T_x$ plus any trees rooted at an $x_i$ 
and another 10 for $y$).
The number of leaves of each $\hT_{x_i}$ now satisfies
$$\hD_\g^{(x_i)}\geq (\frac{4}{\delta^4}\log{n})^{\g+1}\geq n^{\frac{1}{2}+\frac{4}{5}\epsilon}.$$ The result is similar for $\hD_\g^{(y_i)}$.

\noindent Observe next that BFS does not condition on the edges between the leaves $X_i,Y_i$ of the trees $\hT_{x_i}$ 
and $\hT_{y_i}$.
That is, we do not need to look at these edges in order to carry out our construction. 
On the other hand we have conditioned on the occurrence of certain events to imply a certain growth rate.
We handle this technicality as follows. We go through the above construction and halt if ever we find that we cannot
expand by the required amount. Let ${\bf A}$ be the event that we do not halt the construction
i.e. we fail the conditions of Lemmas \ref{lem2} or \ref{lem3}. We have
$\Prob{{\bf A}}=1-o(1)$ and so,

\begin{align*}
\Prob{\exists i:e(X_i,Y_i)=0\mid {\bf A}} &\leq \frac{\Prob{\exists i:e(X_i,Y_i)=0}}{\Pr({\bf A})}   \\ 
\leq 2n^{\frac{4\e}{5}}(1-p)^{n^{1+\frac{8\e}{5}}}  
& \leq n^{-n^\e}.
\end{align*}

\noindent We conclude that {\em whp} there is always an edge between each $X_i,Y_i$ and thus a path of length at most
$(1+2\e)L$ between each $x_i,y_i$.
\end{proof}

\begin{figure}[t]
\centering
\includegraphics[width=0.7\textwidth]{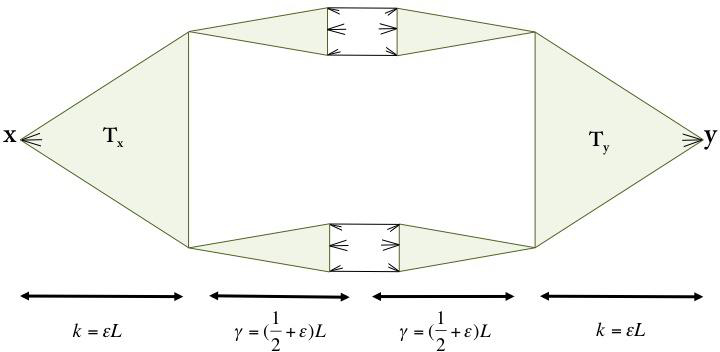} \\
\caption{\label{fig:fig2}  We create edge disjoint paths for each isomorphic pair of leaves $u,\bar{u}$ in the two node disjoint trees $T_x,T_y$ (see Lemma~\ref{lem4} for the details).}
\end{figure}

\noindent Using elementary data structures,  our algorithm runs  in total expected run time $O(n^2(n+m))=O(\frac{n^3\log{n}}{\delta^4})$.    

\subsection{Algorithm Correctness} 
\label{sec:correctness}

Recall from Section~\ref{sec:proposed} that 
$Z_{u, \bar u} :=
\maj( \{ \tau(u, v) Z_{v, \bar v} \tau(\bar v, \bar u) \}_{v \in N(u)} )$. 
Therefore, note that at any level $k$ in the tree,
the random variables $\{Z_{u, \bar u}\}$
are {\bf independent} for all nodes $u$ at level $k$.
(This is true in the base case by path-disjointedness, and preserved by the induction). The key  Lemma~\ref{lem:bias} follows.   In simple terms, it shows that the bias of our estimator improves by roughly a $\delta^2\sqrt{\Delta}$ factor at each level.
 
\noindent 
\begin{lemma}
\label{lem:bias} 
Suppose that for all $v \in T_x$ at depth $k+1$, we have
$$
\Pr[Z_{v, \bar v} = \sigma(v)\sigma(\bar v)] \geq 1/2 + \gamma$$
Then, for all $u \in T_x$ at depth $k$, we have
$$\Pr[Z_{u, \bar u} = \sigma(u)\sigma(\bar u)] \geq 1/2 +
\min(\gamma (c_1 \delta^2 \sqrt{\Delta}), c_2)$$
for some universal $c_1, c_2$.
\end{lemma}

\noindent The proof invokes the Majority Bias Lemma (see Lemma~\ref{lem:boostsbias}) that we prove at the end of this section. 

\begin{proof}
It is more convenient to work with the bias
$$
\Mean{ Z_{v, \bar v} \sigma(v)\sigma(\bar v) }
= 2\Prob{Z_{v, \bar v} = \sigma(v)\sigma(\bar v)} - 1 \geq 2\gamma $$

By the recursive definition,
$$Z_{u, \bar u} :=
\maj( \{ \tau(u, v) Z_{v, \bar v} \tau(\bar v, \bar u) \}_{v \in N(u)} )$$

So:
\begin{align*}
\Mean{Z_{u, \bar u}\sigma(u)\sigma(\bar u)} 
&= \\
\Mean{
\maj( \{ \tau(u, v) Z_{v, \bar v} \tau(\bar v, \bar u) \}_{v \in N(u)} )
\sigma(u)\sigma(\bar u))} &=\\
\Mean{
\maj(
\{ \sigma(u)\tau(u, v) Z_{v, \bar v} \tau(\bar v, \bar u) \sigma(\bar u) \}_{v \in N(u)} )}&=\\
\mathbb{E}[
\maj(
\{ \sigma(u)\sigma(v)\tau(u, v) \cdot
\sigma(v)Z_{v, \bar v} \sigma(\bar v)
\cdot &
\sigma(\bar v)\sigma(\bar u) \tau(\bar v, \bar u)  \}_{v \in N(u)} )] =\\
\mathbb{E}[
\maj(
\{
\eta_{u,v} \cdot
\sigma(v)Z_{v, \bar v} \sigma(\bar v)
\cdot 
\eta_{\bar u, \bar v}
 \}_{v \in N(u)} )] &\\
\end{align*}

For any single $v \in N(u)$ we have
$
\mathbb{E} [\eta_{u, v}]
\mathbb{E}[\sigma(v)Z_{v, \bar v} \sigma(\bar v)]
\mathbb{E}[\eta_{\bar u, \bar v}]
\geq \delta^2 2\gamma
$.
Then by Lemma~\ref{lem:boostsbias}, taking $\mathrm{majority}$ over $\Delta$ such coins amplifies the bias
to
$\min(c_1 \gamma \delta^2 \sqrt{\Delta}, c_2)$, as desired.
\end{proof}

To conclude the analysis, we show in Lemma~\ref{lem:constbias} that doing $\epsilon L$ levels of this amplifies the bias to a constant. Then we are done, because the root will take the majority of $\Delta = \Omega(\log(n)/\delta^4)$ independent coins, each with bias $O(\delta^2)$, and so the estimate is correct with high probability. The following amplification result holds:

\begin{lemma}
\label{lem:constbias}
For nodes $u \in T_x$ that are $\epsilon L$ levels up from the leaves,
we have that $$
\Pr[Z_{v, \bar v} = \sigma(v)\sigma(\bar v)] \geq 1/2 +c_2$$
\end{lemma}

\begin{proof}
Note that at a leaf $u \in T_x$, the bias is
$$
\mathbb{E}[Z_{u, \bar u} \sigma(u)\sigma(\bar u)] 
=
\mathbb{E}[\tau(P_i) \sigma(u)\sigma(\bar u)] 
=
\mathbb{E}[\prod_{e \in P_i}\eta_e]
\geq \delta^{(1+2\epsilon)L}
$$
where $P_i$ is the path from $u \to \bar u$, of length at most $(1+2\epsilon)L$.

Then we apply the amplification lemma inductively for $\epsilon L$ levels, starting with this bias at the leaves.
It suffices to show that
$$(c_1 \delta^2 \sqrt{\Delta})^{\epsilon L} \exp\left(-L (1+2\epsilon) \log\frac{1}{\delta}\right) > 1$$
This means that $\log(c_1 \delta^2 \sqrt{\Delta}) > \frac{1}{\epsilon}(1+2\epsilon) \log \left(\frac{1}{\delta}\right)$. Equivalently, solving for $\Delta$
$$ \Delta > \left(\frac{1}{\delta}\right)^4 \left(\frac{1}{\delta}\right)^{\frac{2+2\epsilon}{\epsilon}},$$
which holds for our choice of $\Delta$ as long as $\delta$ is a constant.
\end{proof}

\begin{lemma}[Majority Bias Lemma]
\label{lem:boostsbias}
Let $X_1, X_2, \ldots X_n$ be independent random variables with $X_i \in \{\pm 1\}$ and $\mathbb{E} X_i \geq \delta$. Then $\mathbb{E} \maj(X_1, \ldots X_n) \geq \min(c_1 \sqrt{n} \delta, c_2)$ for some universal constants $c_1$ and $c_2$.
\end{lemma}
\begin{proof}
First we prove the case when $\mathbb{E} X_i = \delta$ for all $i$. 
Consider the Fourier transform of $\maj(X_1, \ldots, X_n) = \sum_{S \subset[n]} \widehat{\maj_n}(S)\chi_S$, where $\chi_S = \prod_{i \in S}X_i$
and $\widehat{\maj_n}(S)\chi_S$ are the corresponding Fourier coefficients. Specifically, for even $|S|$, $\widehat{\maj_n}(S) = 0$,
and for odd $|S|$,
$$\widehat{\maj_n}(S)   = (-1)^{\tfrac{k-1}{2}} \frac{{\frac{n-1}{2} \choose \frac{k-1}{2}}}{ {n-1 \choose k-1}} \frac{2}{2^n} {n-1 \choose \frac{n-1}{2}}.$$

 Then 
\begin{align*}
\mathbb{E} [\maj(X_1, \ldots X_n)] & =  \sum_S \widehat{\maj_n}(S) \mathbb{E} [\chi_S] \\
& =  \sum_S\widehat{\maj_n}(S)\delta^{|S|} \\
& \geq  \sum_{|S| = 1}\widehat{\maj_n}(S)\delta - \left|\sum_{|S| \geq 2}\widehat{\maj_n}(S)\delta^{|S|}\right| \\
& \geq \sum_{|S| = 1}\widehat{\maj_n}(S)\delta - \sum_{k \geq 2} \delta^k\left(\sum_{|S| = k}|\widehat{\maj_n}(S)|\right).
\end{align*}
And we have for $|S| = 1$ 
\begin{equation*}
\sum_{|S|=1}\widehat{\maj_n}(S) = \frac{2}{2^n}{n-1 \choose \frac{n-1}{2}} = \frac{2\sqrt{2}}{\sqrt{\pi}}\sqrt{n}
\end{equation*}
For $|S| = k \geq 2$ using the Cauchy-Schwarz inequality we can provide an upper bound
\begin{equation*}
\sum_{|S|=k}|\widehat{\maj_n}(S)| \leq \sqrt{{n \choose k}}\sqrt{\sum_{|S| = k} \widehat{\maj_n}(S)^2},
\end{equation*}
 and Parseval's identity implies $\sum_{|S| = k} \widehat{\maj_n}(S)^2  \leq \sum_{S \subset [n]} \widehat{\maj_n}(S)^2 = 1$, which ultimately gives an upper bound $\sum_{|S| = k} |\widehat{\maj_n}(S)| \leq n^{k/2}$.

Plugging those two together, when $\delta\sqrt{n} \leq \frac{1}{2}$, we have
\begin{equation*}
\mathbb{E} [\maj(X_1, \ldots X_n)] \geq \frac{2\sqrt{2}}{\sqrt{\pi}}\sqrt{n}\delta - \sum_{k\geq 2}(\sqrt{n}\delta)^k = \Omega(\sqrt{n}\delta).
\end{equation*}

When $\delta\sqrt{n} > \frac{1}{2}$, by Chernoff bound
\begin{equation*}
\mathbb{E} [\maj(X_1, \ldots X_n)] = \Pr[\sum_{i}X_i > 0] \geq 1 - e^{-\Omega(\delta^2n)} = \Omega(1)
\end{equation*}
This gives $\mathbb{E} [\maj(X_1, \ldots X_n)] \geq \min(c_1 \sqrt{n} \delta, c_2)$ for $\mathbb{E} [X_i] = \delta$. 
The general case when $\mathbb{E} [X_i]\geq \delta$ follows readily from the monotonicity of the majority function.

 \end{proof}

\section{Conclusion}
\label{sec:concl}
An interesting  open problem concerns the extension of our results to $k$ clusters. 
Specifically, our clustering model naturally extends to the case where there are more than two clusters \cite{mazumdar2017clustering}. In this case the set $V$ of 
$n$ items belong to $k$ clusters.
When we query the pair of nodes $\{u,v\}$ we obtain a noisy answer on whether $u,v$ belong to the same cluster or not. 
Can we design a polynomial time algorithm that performs $O(\frac{kn \log n}{\delta^2})$ queries for all $0<\delta<1$?  From an experimental point of view,  we plan to experiment with other types of classifiers in addition to logistic regression classifiers to diagnose whether an additional improvement in classification accuracy can be achieved.

\bibliographystyle{abbrv}
\bibliography{ref} 

\begin{thebibliography}{10}

\bibitem{snap}
Stanford network analysis project, July 2017.
\newblock \url{http://snap.stanford.edu/data/index.html}.

\bibitem{abbe2016exact}
E.~Abbe, A.~S. Bandeira, and G.~Hall.
\newblock Exact recovery in the stochastic block model.
\newblock {\em IEEE Transactions on Information Theory}, 62(1):471--487, 2016.

\bibitem{bansal2004correlation}
N.~Bansal, A.~Blum, and S.~Chawla.
\newblock Correlation clustering.
\newblock {\em Machine Learning}, 56(1-3):89--113, 2004.

\bibitem{bollobas1998random}
B.~Bollob{\'a}s.
\newblock Random graphs.
\newblock In {\em Modern Graph Theory}. Springer, 1998.

\bibitem{bonchi2014correlation}
F.~Bonchi, D.~Garcia-Soriano, and E.~Liberty.
\newblock Correlation clustering: from theory to practice.
\newblock In {\em KDD}, page 1972, 2014.

\bibitem{broder1998optimal}
A.~Z. Broder, A.~M. Frieze, S.~Suen, and E.~Upfal.
\newblock Optimal construction of edge-disjoint paths in random graphs.
\newblock {\em SIAM Journal on Computing}, 28(2):541--573, 1998.

\bibitem{candes2006robust}
E.~J. Cand{\`e}s, J.~Romberg, and T.~Tao.
\newblock Robust uncertainty principles: Exact signal reconstruction from
  highly incomplete frequency information.
\newblock {\em IEEE Transactions on information theory}, 52(2):489--509, 2006.

\bibitem{cartwright1956structural}
D.~Cartwright and F.~Harary.
\newblock Structural balance: a generalization of heider's theory.
\newblock {\em Psychological review}, 63(5):277, 1956.

\bibitem{cesa2012correlation}
N.~Cesa-Bianchi, C.~Gentile, F.~Vitale, G.~Zappella, et~al.
\newblock A correlation clustering approach to link classification in signed
  networks.
\newblock In {\em COLT}, pages 34--1, 2012.

\bibitem{chen2014clustering}
Y.~Chen, A.~Jalali, S.~Sanghavi, and H.~Xu.
\newblock Clustering partially observed graphs via convex optimization.
\newblock {\em Journal of Machine Learning Research}, 15(1):2213--2238, 2014.

\bibitem{chen2012clustering}
Y.~Chen, S.~Sanghavi, and H.~Xu.
\newblock Clustering sparse graphs.
\newblock In {\em Advances in neural information processing systems}, pages
  2204--2212, 2012.

\bibitem{chiang2014prediction}
K.-Y. Chiang, C.-J. Hsieh, N.~Natarajan, I.~S. Dhillon, and A.~Tewari.
\newblock Prediction and clustering in signed networks: a local to global
  perspective.
\newblock {\em Journal of Machine Learning Research}, 15(1):1177--1213, 2014.

\bibitem{chiang2011exploiting}
K.-Y. Chiang, N.~Natarajan, A.~Tewari, and I.~S. Dhillon.
\newblock Exploiting longer cycles for link prediction in signed networks.
\newblock In {\em Proceedings of the 20th ACM international conference on
  Information and knowledge management}, pages 1157--1162. ACM, 2011.

\bibitem{dudek2015rainbow}
A.~Dudek, A.~M. Frieze, and C.~E. Tsourakakis.
\newblock Rainbow connection of random regular graphs.
\newblock {\em SIAM Journal on Discrete Mathematics}, 29(4):2255--2266, 2015.

\bibitem{easley2010networks}
D.~Easley and J.~Kleinberg.
\newblock {\em Networks, crowds, and markets: Reasoning about a highly
  connected world}.
\newblock Cambridge University Press, 2010.

\bibitem{frieze2012rainbow}
A.~Frieze and C.~E. Tsourakakis.
\newblock Rainbow connectivity of sparse random graphs.
\newblock In {\em Approximation, Randomization, and Combinatorial Optimization
  (APPROX-RANDOM)}, pages 541--552. Springer, 2012.

\bibitem{hajek2016achieving}
B.~Hajek, Y.~Wu, and J.~Xu.
\newblock Achieving exact cluster recovery threshold via semidefinite
  programming.
\newblock {\em IEEE Transactions on Information Theory}, 62(5):2788--2797,
  2016.

\bibitem{harary1953notion}
F.~Harary.
\newblock On the notion of balance of a signed graph.
\newblock {\em The Michigan Mathematical Journal}, 2(2):143--146, 1953.

\bibitem{heider1946attitudes}
F.~Heider.
\newblock Attitudes and cognitive organization.
\newblock {\em The Journal of psychology}, 21(1):107--112, 1946.

\bibitem{hou2016new}
J.~P. Hou, A.~Emad, G.~J. Puleo, J.~Ma, and O.~Milenkovic.
\newblock A new correlation clustering method for cancer mutation analysis.
\newblock {\em arXiv preprint arXiv:1601.06476}, 2016.

\bibitem{itai1982complexity}
A.~Itai, Y.~Perl, and Y.~Shiloach.
\newblock The complexity of finding maximum disjoint paths with length
  constraints.
\newblock {\em Networks}, 12(3):277--286, 1982.

\bibitem{leskovec2010predicting}
J.~Leskovec, D.~Huttenlocher, and J.~Kleinberg.
\newblock Predicting positive and negative links in online social networks.
\newblock In {\em Proceedings of the 19th international conference on World
  Wide Web (WWW)}, pages 641--650. ACM, 2010.

\bibitem{leskovec2010signed}
J.~Leskovec, D.~Huttenlocher, and J.~Kleinberg.
\newblock Signed networks in social media.
\newblock In {\em Proceedings of the SIGCHI conference on Human Factors in
  Computing Systems}, pages 1361--1370. ACM, 2010.

\bibitem{makarychev2015correlation}
K.~Makarychev, Y.~Makarychev, and A.~Vijayaraghavan.
\newblock Correlation clustering with noisy partial information.
\newblock In {\em Proceedings of the Conference on Learning Theory (COLT)},
  volume~6, page~12, 2015.

\bibitem{mathieu2010correlation}
C.~Mathieu and W.~Schudy.
\newblock Correlation clustering with noisy input.
\newblock In {\em Proceedings of the twenty-first annual ACM-SIAM symposium on
  Discrete Algorithms}, pages 712--728. Society for Industrial and Applied
  Mathematics, 2010.

\bibitem{mazumdar2016clustering}
A.~Mazumdar and B.~Saha.
\newblock Clustering via crowdsourcing.
\newblock {\em arXiv preprint arXiv:1604.01839}, 2016.

\bibitem{mazumdar2017clustering}
A.~Mazumdar and B.~Saha.
\newblock Clustering with noisy queries.
\newblock In {\em Advances in Neural Information Processing Systems}, pages
  5790--5801, 2017.

\bibitem{mcsherry2001spectral}
F.~McSherry.
\newblock Spectral partitioning of random graphs.
\newblock In {\em Proceedings. 42nd IEEE Symposium on Foundations of Computer
  Science (FOCS)}, pages 529--537. IEEE, 2001.

\bibitem{mitzenmacher2016predicting}
M.~Mitzenmacher and C.~E. Tsourakakis.
\newblock Predicting signed edges with $ o (n^{1+o(1)}) $ queries.
\newblock {\em arXiv preprint arXiv:1609.00750}, 2016.

\bibitem{mitzenmacher2005probability}
M.~Mitzenmacher and E.~Upfal.
\newblock {\em Probability and computing: Randomized algorithms and
  probabilistic analysis}.
\newblock Cambridge university press, 2005.

\bibitem{o2014analysis}
R.~O'Donnell.
\newblock {\em Analysis of boolean functions}.
\newblock Cambridge University Press, 2014.

\bibitem{pedregosa2011scikit}
F.~Pedregosa, G.~Varoquaux, A.~Gramfort, V.~Michel, B.~Thirion, O.~Grisel,
  M.~Blondel, P.~Prettenhofer, R.~Weiss, V.~Dubourg, et~al.
\newblock Scikit-learn: Machine learning in python.
\newblock {\em Journal of Machine Learning Research}, 12(Oct):2825--2830, 2011.

\bibitem{konect:kenneth54}
K.~E. Read.
\newblock Cultures of the {Central} {Highlands}, {New} {Guinea}.
\newblock {\em Southwestern J. of Anthropology}, 10(1):1--43, 1954.

\bibitem{settles2010active}
B.~Settles.
\newblock Active learning literature survey.
\newblock {\em University of Wisconsin, Madison}, 52(55-66):11, 2010.

\bibitem{shamir2004cluster}
R.~Shamir, R.~Sharan, and D.~Tsur.
\newblock Cluster graph modification problems.
\newblock {\em Discrete Applied Mathematics}, 144(1):173--182, 2004.

\bibitem{enemy}
S.~Strogatz.
\newblock The enemy of my enemy, February 2014.
\newblock
  \url{https://opinionator.blogs.nytimes.com/2010/02/14/the-enemy-of-my-enemy/}.

\bibitem{tsourakakis2013mathematical}
C.~E. Tsourakakis.
\newblock {\em Mathematical and Algorithmic Analysis of Network and Biological
  Data}.
\newblock PhD thesis, Carnegie Mellon University, 2013.

\bibitem{verroios2015entity}
V.~Verroios and H.~Garcia-Molina.
\newblock Entity resolution with crowd errors.
\newblock In {\em IEEE 31st International Conference on Data Engineering
  (ICDE)}, pages 219--230. IEEE, 2015.

\bibitem{vu2014simple}
V.~Vu.
\newblock A simple svd algorithm for finding hidden partitions.
\newblock {\em arXiv preprint arXiv:1404.3918}, 2014.

\end{thebibliography}

\end{document}